\documentclass[12pt]{article}
\usepackage{hyperref}
\usepackage{mathrsfs,bm,amssymb,color}
\usepackage{latexsym,amscd}
\usepackage[sf,bf,tiny,center]{titlesec}
\usepackage{fancyhdr}
\usepackage{float}
\usepackage{hyperref}
\usepackage{color}
\usepackage{booktabs}
\usepackage{multirow}
\usepackage{graphicx}
\usepackage{url}
\usepackage{color,hyperref}

\usepackage{graphicx,amsmath,amsfonts,latexsym,amssymb,amsthm,amscd}
\usepackage{latexsym,amscd}
\newcommand{\eq}[1]{\begin{equation}\label{#1}}
\newcommand{\en}{\end{equation}}
\newcommand{\eqst}[1]{\begin{equation*}\label{#1}}
\newcommand{\enst}{\end{equation*}}
\newcommand{\eqr}[1]{\begin{eqnarray}\label{#1}}
\newcommand{\enr}{\end{eqnarray}}
\newcommand{\eqrst}[1]{\begin{eqnarray*}\label{#1}}
\newcommand{\enrst}{\end{eqnarray*}}

\makeatletter
\@addtoreset{equation}{chapter}
\makeatother

\numberwithin{equation}{section}

\newcommand{\ep}{\end{proposition}}
\newcommand{\bc}[1]{\begin{corollary}\label{#1}}
\newcommand{\ec}{\end{corollary}}
\newcommand{\bdf}[1]{\begin{definition}\label{\rm #1}}
\newcommand{\edf}{\end{definition}}
\newcommand{\bt}[1]{\begin{theorem}\label{#1}}
\newcommand{\et}{\end{theorem}}
\newcommand{\bl}[1]{\begin{lemma}\label{#1}}
\newcommand{\el}{\end{lemma}}
\newcommand{\bp}[1]{\begin{proposition}\label{#1}}
\newcommand{\br}[1]{\begin{remark}\label{#1}}
\newcommand{\er}{\end{remark}}
\newcommand{\bi}{\begin{description}}
\newcommand{\ei}{\end{description} }

  \newcommand{\beq}{\begin{equation}}
  \newcommand{\eeq}{\end{equation}}

  \newtheorem{theorem}{Theorem}[section]
\newtheorem{definition}[theorem]{Definition}
\newtheorem{remark}[theorem]{Remark}
  \newtheorem{lemma}[theorem]{Lemma}

\newtheorem{proposition}[theorem]{Proposition}

\newtheorem{corollary}[theorem]{Corollary}

\begin{document}
\begin{center}
\large{\textbf{A sensitivity analysis of a gonorrhoea dynamics and 
control model}}
\end{center}

\begin{center}
Louis Omenyi$^{1,3,\ast , 
\href{https://orcid.org/0000-0002-8628-0298}{id}},$
 Aloysius Ezaka$^{2},$ Henry O. Adagba$^{2},$  Gerald Ozoigbo$^{1},$ 
 Kafayat Elebute$^{1},$ \\
$^{1}$ Department of Mathematics and Statistics,  \\
Alex Ekwueme Federal University, Ndufu-Alike,  Nigeria\\
$^{2}$Department of Industrial Mathematics and Applied Statistics, \\
Ebonyi State University, Abakaliki, Nigeria\\
$^{3}$Department  of  Mathematical Sciences,  \\
Loughborough University, Leicestershire, United Kingdom\\ 
 Corresponding author, email: \url{omenyi.louis@funai.edu.ng} 
\end{center}

\begin{abstract}
We formulate and analyse a robust mathematical model of the dynamics 
of gonorrhoea incorporating passive immunity and control.
Our results show that the disease-free  and endemic equilibria 
of the model are both locally and globally asymptotically stable. 
A sensitivity analysis of the model shows that the dynamics of the 
model is variable and dependent on waning rate, control parameters and 
interaction of the latent and infected classes. In particular, the 
lower the waning rate, the more the exponential decrease in the passive 
immunity but the susceptible population increases  to the equilibrium 
and wanes asymptotically due to the presence of the control parameters 
and restricted interaction of the latent and infected classes.
\end{abstract}

\section{Introduction}
Gonorrhoea, a sexually transmitted disease, is a major cause of  
infertility and other debilitating ailments among couples. Thus, 
it becomes necessary to undertake  prompt prevention and control 
activities to tackle the ugly incidence of this sexually transmitted 
diseases, \cite{Gre}. It is caused by a bacterium called Neisseria 
gonorrhoeae, \cite{Une}. According to  Mushayabasa and Bhunu in 
\cite{Mus2}, neisseria gonorrhoea is characterized by a very short 
period of latency, namely, $2 -10$ days and is commonly found in the 
glummer epithelium such as the urethra and endo-cervix epithelia of the 
reproductive track, \cite{Gar}. Gonorrhoea is transmitted to a new born 
infant from the infected mother through the birth canal thereby causing 
inflammations and eye infection such as conjunctivitis. It 
is also spread through unprotected sexual intercourse, \cite{Ugw}.

Studies by Usman and Adam \cite{Usm} and Center for Disease Control 
Report in \cite{CDC} show that male patients of gonorrhoea have  
pains in the testicles (known as epididymitis), painful urination due 
to scaring inside the urethra while in female patients, the disease 
may  ascend the genital tract and block the fallopean tube  leading to 
pelvic inflammatary disease (PID)  and infertility, see also 
\cite{Ram}. Other complications associated with this epidemic include 
arthritis, endocarditis, chronic pelvic pain, meningitis and ectopic 
pregnancy, \cite{Ril}. 

Gonorrhoea confers temporal immunity on some individuals in the 
susceptible class while some others are not immuned, \cite{Ugw}. This 
immunity through the immune system plays an important role in 
protecting the body against the infection and other foreign substances, 
\cite{CDC}. That is why an immuno-compromised 
patient has a reduced ability to fight infectious disease such as 
gonorrhoea due to certain diseases and genetic disorder, \cite{Sch}. 
Such patient may be particularly vulnerable to opportunistic infection 
such as gonorrhoea. Hence, immune reaction can be stimulated by 
drug-induced immune system such as Thrombocytopenia, \cite{Sch}. 
This helps to reduce the waning rate of passive immunity in the immune 
class, \cite{Bas}. However, if the activity of 
immune system is excessive or over-reactive due to lack of cell 
mediated immunity, a hypersensitive reaction develops such as 
auto-immunity and allergy which may be injurious to the human body 
or may even cause death \cite{WHO}.

Statistically, gonorrhoea infection has spread worldwide with more 
than $360$ million new cases witnessed globally in adults aged 
$15-49$ years, \cite{CDC}.  In 1999, above over $120$ million 
people in African countries were reported to have contracted the 
disease. Over $82$ million people were reported in Nigeria, 
\cite{CDC}. Researches abound on the modelling and control of this 
epidemic with various approaches and controls, see e.g. 
\cite{CDC, Jin, Mus1, Sha, Sac, Ugw, Une} and mostly 
recently \cite{Ibr, Whi, Osn} and \cite{Did}. This present study 
continues the discussion by investigating the dynamics of the 
disease by carrying out the sensitivity analysis of the effective 
reproduction number of the model.  Besides, the theory of epidemiology 
signifies the phenomenon of bifurcation at the equilibria. This is a 
classical requirement for the model's effective reproduction number, 
$R_{e}.$  Although this is necessary, it is no longer sufficient to 
conclude the effective control or elimination of gonorrhoea in a 
population, see e.g. \cite{WHO}. Therefore in this paper, we consider 
the nature of the equilibrium solution near the bifurcation point 
$R_{e} = 1$ in the neighbourhood of the disease-free equilibrium,  
$E_{0},$ through sensitivity analysis to determine the most sensitive 
parameters to target by a way of intervention strategy.
 
\section{Materials and Methods}
We formulate a modified SIR model by extending an existing one to 
incorporate passive immunity. Let $Q(t)$ be passive immune class, 
$S(t)$ the susceptible compartment, $L(t)$ the  latent class,   
$I(t)$ the infectious class, $T(t)$ the  treated class and  $R(t)$ 
be the recovered compartment over time $t.$ Let the parameters of the 
model be $\sigma$  as level of  recruitment, $\upsilon$ as waning rate 
of immunity, $\mu$  as rate of natural mortality, $\lambda$ as contact 
rate  between the susceptible and the latent classes, $\eta$ as  
treatment rate of latent class, $\gamma$ as induced death rate due to 
the infection, $\alpha$ as treatment rate of infected compartment, 
$\beta$ as  infectious rate of latent class, $\omega$ as recovery rate 
of treated class, $\delta$ as rate at which recovered class becomes 
susceptible again, $\theta$ as infectious rate from the susceptible 
class direct to the infectious class, $k_{1}$ as control measure given 
to latent class and $k_{2}$ as control measure given to infected class. 

Next, assume the recruitment into the population is by birth or 
immigration; all the parameters of the model are positive; some 
proportions of new birth are immunized against the 
infection; the immunity conferred on the new birth wanes after 
sometime; and that the rate of contact of the disease due to 
interaction, $\lambda ,$ rate is due to the movement of the infected 
population. Consequently, the total population at time $t$  is 
$N(t) = Q(t) + S(t) + L(t) + I(t) + T(t) + R(t).$ The existing 
schematic model diagram is given as Figure \ref{model1A} and that of 
the extended model as \ref{model2}: 
\begin{figure}[h!]
\begin{center}
\includegraphics[width=0.6\textwidth]{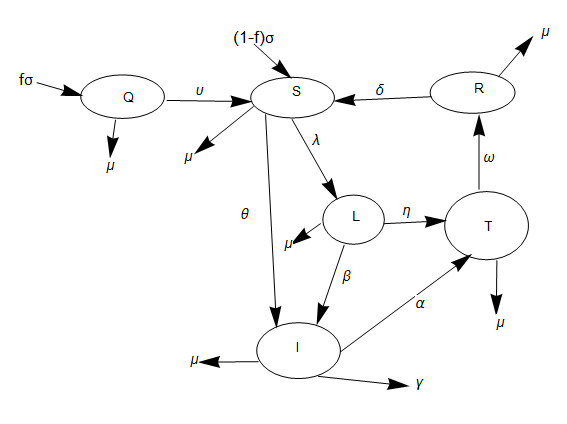}
\caption{The diagram of the  model; see \cite{Ugw}.}
\label{model1A} 
\end{center}
\end{figure}

\begin{figure}[!ht]
\begin{center}
\includegraphics[width=0.6\textwidth]{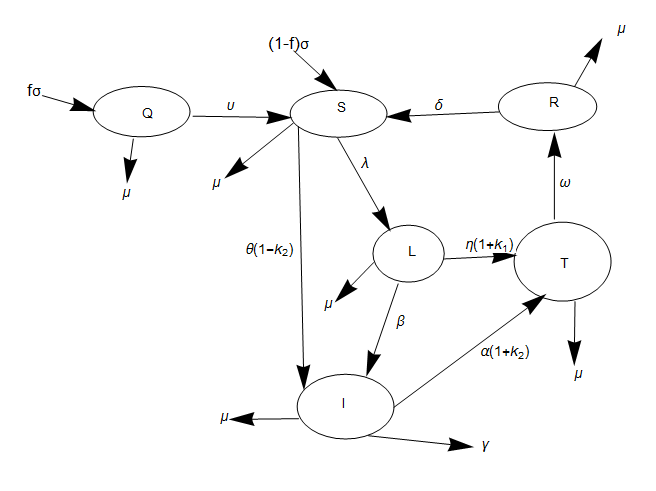}
\caption{The  model with control}
\label{model2} 
\end{center}
\end{figure}

So, the system of equations for the model without control is 
\eqref{3.9}:  
\begin{eqnarray}\label{3.9} 
\left. \begin{array}{rcl}
\frac{dQ}{dt} & =& f \sigma - \upsilon Q - \mu Q \\
\frac{dS}{dt} & =& \upsilon Q + (1 - f) \sigma + \delta R - \mu S - 
\theta S - \theta S I \\
\frac{dL}{dt} & =& \theta SI - \eta L - \mu L - \beta L \\
\frac{dI}{dt} & =& \beta L + \theta S - \mu I - \gamma I - \alpha I \\
\frac{dR}{dt} & =& \omega T - \mu R -  \delta R \\
\frac{dT}{dt} & =& \alpha I + \eta L - \mu T - \omega T
\end{array} \right\}
\end{eqnarray}
and those with control is \eqref{3.9b}:  
\begin{eqnarray} \label{3.9b}
\left. \begin{array}{rcl}
\frac{dQ}{dt} & =& f \sigma - \upsilon  Q -\mu Q \\
\frac{dS}{dt} & =& \upsilon  Q + (1 - f) \sigma - \theta S(1 - k_{2} ) 
+  \delta R - \mu S - \theta S I \\
\frac{dL}{dt} & =& \theta S I - \beta L - \mu L - \eta(1 + k_{1} )L \\
\frac{dI}{dt} & =& \beta L  + \theta S (1-k_{2}) - ((\mu + \gamma)  +  
\alpha(1 +  k_{2} ))I \\
\frac{dR}{dt} & =& \omega T - \mu R - \delta R \\
\frac{dT}{dt} & =& \eta(1 + k_{1} )L + \alpha(1 + k_{2} )I - \mu T -
\omega T
\end{array} \right\}
\end{eqnarray}

By sensitivity theory, see e.g. \cite{Bas, Gar, Jin, Sha} and  
\cite{Om19}, it is expected that a significant perturbation in the 
model parameters leads to a change in the behaviour of the equilibrium 
solution of the model. We proceed to study the qualitative properties 
of the model through the variation of the model parameters such as in 
\cite{Gre} and \cite{Hoo}.

\section{Results and discussions}
We firstly prove the positivity of the solution set of the model. 
\begin{lemma} \label{T1} 
The extended gonorrhoea model equations admit non-negative 
solution set.
\end{lemma}

\begin{proof}
A direct computation of the model equations (\ref{3.9b}) gives  
\begin{eqnarray*}
Q(t) & = & \frac{f \sigma}{\mu + \upsilon} + c_{1} 
e^{-(\mu + \upsilon)t} > 0 ; \\ 
 S(t) & =& \frac{\upsilon Q + (1 - f)\sigma + \delta R}{\mu + \theta 
 (1 - k_{2} ) +\theta I} + c_{2} e^{-(\mu + \theta (1 - k_{2} ) 
 + \theta I)t}  > 0 \Rightarrow 
 S(t) > 0; \\
L(t) & =& \frac{\theta IS}{\mu + \beta + \eta(1 + k_{1} )} +
 c_{3} e^{-(\mu + \beta + \eta(1 + k_{1} ))t}  > 0 \Rightarrow 
 L(t) > 0; \\
  I(t) & =& \frac{\beta L + \theta S(1 - k_{2} )}{\mu + 
\gamma + \alpha(1 + k_{2} )} +
 c_{4} e^{-(\mu + \gamma + \alpha(1 + k_{2} ))t}  > 0 ; \\
R(t) & =& \frac{\omega T}{\mu + \delta} + c_{5} e^{-(\mu + \delta)t}  
> 0 \Rightarrow 
R(t) > 0 ~~ \text{and} ~~ \\
T(t) & =& \frac{\eta(1 + k_{1} )L 
+ \alpha(1 + k_{2} )I}{\mu + \omega} + 
c_{6} e^{-(\mu + \omega)t} > 0 \Rightarrow   
T(t) > 0
\end{eqnarray*}
for arbitrary positive constants $c_{1}, c_{2}, \cdots , c_{6}.$
\end{proof}

Next, we show that there exists a feasible solutions region for the 
model.
 
\begin{proposition}\label{T2}
Let $\dot{x}  =  f(x)$   in   $D \subset  \mathbb{R}_{+}^{n}$ be 
a system of equations in the feasible solutions region of the model, 
\eqref{3.9b}.  Then the  solutions  $x(t)$  
are  feasible  for  all  $t \geq 0$  if   
$x(t)  \in D  \subset  \mathbb{R}_{+}^{6} .$ 
\end{proposition}

\begin{proof}
It suffices to prove that the solution set
$x(t) = \{( Q(t), S(t), L(t), I(t), T(t), R(t))\}$ enters the 
 invariant region  $ \mathbb{D} \subset \mathbb{R}_{+} ^6 .$  
Since $N = Q + S + L + I + R + T,$ it follows that 
 \[\frac{d N}{d t} =  \frac{d Q}{d t} + \frac{d S}{d t} 
 + \frac{d L}{d t} + \frac{d I}{d t} + 
 \frac{d R}{d t} + \frac{d T}{d t} .\]
This implies that  
\[\frac{dN}{dt} = \sigma - \mu N - \gamma I \leq \sigma - \mu N \\
\Rightarrow  \frac{dN}{dt}  \leq  \sigma - \mu N .\]
Solving gives   
\[N \leq \frac{\sigma - q e^{-\mu t}}{\mu} ~~ \text{with} ~~
N(0) = N0 \leq \frac{\sigma - q e^{-\mu t}}{\mu} \]
where $q$ is arbitrary constant and $q = \sigma - \mu N(0).$  
Therefore, 
\[N(t) \leq \frac{\sigma}{\mu} - \frac{(\sigma - \mu N(0)) 
e ^{-\mu t}}{\mu} .\]
Observe that as  $t \rightarrow \infty,$ the population size approaches 
the carrying capacity $\frac{\sigma}{\mu}.$ 
That is, 
\[0 < N \leq \frac{\sigma}{\mu} \Rightarrow N \rightarrow 
\frac{\sigma}{\mu}. \]
 Hence, the feasible solution set of the system enters the invariant 
 region  $D \subset \mathbb{R}^{6} _{+}$, and 
 \[S > 0, Q > 0, L \geq  0, I \geq 0, T \geq 0, R \geq 0 .\]
Whenever $N \leq \frac{\sigma}{\mu}$, every solution with initial 
condition in $D$ remains in that region for $t> 0.$ Hence, the region 
$D$ is positively invariant or bounded with $\frac{\sigma}{\mu}$ 
as the upper bound, $x(t) \in D.$
\end{proof} 

\begin{lemma} 
The extended model admits disease-free and endemic  equilibria.
\end{lemma} 
 
\begin{proof}
Setting the right hand side of the system \eqref{3.9b} 
equal to zero supposes there is no gonorrhoea infection in
the population, i.e. $L = I = R = T = 0$ which gives the 
disease-free equilibria. On the other hand, suppose   
$L \neq 0, ~ I \neq 0,~ R \neq 0 ~~ \text{and} ~~ T \neq 0,$ 
one solves to obtain the endemic equilibrium state of the model 
to be 
\begin{eqnarray*}
Q^{\ast} & =& \frac{f \sigma}{\mu + \upsilon}; \\
S^{\ast} & =& \frac{(\mu + \delta)(\mu + \omega) f \sigma 
+ (\mu + \upsilon)(\mu + \delta)(\mu + \omega) 
\sigma(1 - f) + (\mu + \upsilon) \delta 
\omega(\alpha + \eta)}{(\mu + \upsilon)(\mu 
+ \delta)(\mu + \omega )} ;\\
L^{\ast} & =& \frac{(\lambda)(\mu + \delta)(\mu + \omega) f \sigma 
+ (\mu + \upsilon)(\mu + \delta)(\mu + \omega) \sigma(1 - f) 
+ (\mu + \upsilon) \delta \omega(\alpha + \eta) }{(\mu +
 \beta + \eta)(\mu + \upsilon)(\mu + 
\delta)(\mu + \omega ) } ;\\
I^{\ast} & =& \frac{(\mu + \delta)(\mu + \omega) f \sigma 
+ (\mu + \upsilon)(\mu + \delta)(\mu + \omega) \sigma(1 - f) 
+ (\mu + \upsilon) \delta \omega(\alpha + \eta)(\beta \lambda 
+ ( \mu + \beta + \eta) \theta )}{(\mu + \alpha + \gamma)(\mu 
+ \beta + \eta)(\mu + \upsilon)(\mu + \delta)(\mu + \omega )} ;\\
R^{\ast} & =& \frac{\omega(\alpha + \eta)}{(\mu + \delta)(\mu + \omega)} ~~ \text{and} ~~
T^{\ast} = \frac{\alpha + \eta}{\mu + \omega}.
\end{eqnarray*}
\end{proof}

Now, we recall that the basic reproduction number  $R_{0}$ is 
the expected number of secondary infection produced in a completely 
susceptible population by a typical or one infected individual. 
It determines  how long an infectious disease  prevails 
in a given population, see e.g. \cite{Van}. When $R_{0} < 1,$ 
it indicates that with time the  disease will die out of the 
population thereby giving it a clean health bill. But if $R_{0}$ 
is greater than one, we expect the disease will persist in the 
population. So for the infection to die out of the population, 
$R_{0}$ must be less than $1.$  We have the following result.

\begin{theorem}
The basic reproduction number of the extended model 
$0 < R_{0} \ll 1 .$
\end{theorem}

\begin{proof}
Following  \cite{Het}, we obtain the basic reproduction number 
$R_{0}$ from the model \eqref{3.9} without control measures,  using next generation matrix method.

So,  let 
\begin{equation*}
f_{i} = \begin{bmatrix}
\theta I S \\
\theta S\\
0
\end{bmatrix}
\end{equation*}
so that 
\begin{equation*}
\frac{\partial f_{i}}{\partial x_{j}}
 =  F= \begin{pmatrix}
0 & \theta S & 0 \\
0 & 0 & 0 \\
0 & 0 & 0
\end{pmatrix};
\end{equation*}
and let 
\begin{equation*}
v_{i} = \begin{pmatrix}
\beta L + \mu L + \eta L \\
\mu I + \gamma I + \alpha I - \beta L - \theta S \\
\mu T + \omega T - \eta L - \alpha I
\end{pmatrix} 
\end{equation*}
so that  
\begin{equation*}
\frac{\partial v_{i}}{\partial x_{j}}  =  \begin{pmatrix}
(\beta + \mu + \eta) & 0 & 0 \\
- \beta & (\mu + \gamma + \alpha) & 0 \\
- \eta & - \alpha & (\mu + \omega)
\end{pmatrix}.
\end{equation*}
The matrix formed by the co-factors of the determinant is
\begin{equation*}
V= \begin{pmatrix}
(\mu + \gamma + \alpha)(\mu + \omega) & \beta(\mu + \omega) &
\alpha \beta + \eta(\mu + \gamma + \alpha ) \\
0 & (\beta + \mu + \eta )(\mu + \omega) &  \alpha(\beta
+ \mu + \eta) \\
0 & 0 & (\beta + \mu + \eta)(\mu + \gamma + \alpha)
\end{pmatrix}
\end{equation*}
with inverse 
\begin{equation*}
V^{-1} = \begin{pmatrix}
\frac{1}{\beta + \mu + \eta} & 0 & 0 \\
\frac{\beta}{(\beta + \mu + \eta)(\mu + \alpha +
\gamma)} & \frac{1}{\mu + \gamma + \alpha} & 0 \\
\frac{\alpha \beta + \eta(\mu + \alpha + \gamma)}{ (\mu + \alpha 
+ \gamma)
(\mu + \omega)}  & \frac{\alpha}{\mu + \gamma 
+ \alpha )(\mu + \omega)} & \frac{1}{\mu + \omega}
\end{pmatrix}.
\end{equation*}

Thus, 
\begin{equation*}
|F V^{-1} - \lambda I| = 
\begin{vmatrix}
\frac{\beta \theta S}{(\beta + \mu + \eta)(
\mu + \gamma + \alpha )}  - \lambda &  \frac{\theta S}
{\mu + \gamma + \alpha} & 0 \\
0 & 0 - \lambda & 0 \\
0 & 0 & 0 - \lambda
\end{vmatrix} = 0
\end{equation*}
implies  
\begin{equation*}
\lambda^{2} (\frac{(\beta \theta S)}{(\beta + \mu + \eta)
(\mu + \gamma + \alpha)} - \lambda ) = 0.
\end{equation*}
Therefore, either 
$\displaystyle{\lambda^{2} = 0}$ or
$\displaystyle{\lambda = \frac{(\beta \theta S)}{(\beta + \mu + \eta)
(\mu + \gamma + \alpha)}}. $
So for $\lambda \neq 0$ we obtain the $R_{0}$ to be  
\begin{equation}\label{repno}
 R_{0} = \frac{(\beta \theta S)}{(\beta + \mu 
 + \eta)(\mu + \gamma + \alpha)} 
 = \frac{\sigma \beta \theta(\mu + \upsilon) - \mu f)}{\mu (\mu 
+ \alpha + \gamma )(\mu + \beta + \eta ) (\mu + \upsilon)}  .
\end{equation}
\begin{table}[!!ht]
\centering
  \resizebox{0.8\textwidth}{!}{\begin{minipage}{\textwidth}
 \centering
\begin{tabular}{|l|c|c|c|c|c|c|c|c|c|c|r|}
\hline 
Parameter/Variable  & $\beta$  &  $\theta$ & $\mu $ & $\eta $ & 
$\gamma $ &  $\alpha $ &  $\delta $ &  $\upsilon $ & $\omega $  & 
$\sigma $ \\
\hline
Value & $0.01$ & $0.5$ &  $0.2$ & $0.1$ &  $0.01$ &  $0.2$ & $0.8$ & 
$0.4$ &   $0.7$ &   $0.4$ \\
\hline
Parameter/Variable & $d_1= k_{1} $ &  $d_2=k_{2} $ &   
$f$ & $S$  &  $Q$ &  $R$  &  $T$ &   $L$ &  $I$ & \\
\hline
Value  & $0.5$ & $0.8$ &   $0.91$ & $2000$ & $1000$ & 
$500$ & $1000$ & $1000$ & $500$ &  \\
\hline 
\end{tabular}
\end{minipage}}
\caption{Parameters/variables and values.}
\label{dt}
\end{table}\\
Using the data Table \ref{dt}, we quickly observe that 
\[ R_{0}  = 0.1566070960 \ll 1. \]
When control measure is given to a model, we compute the effective 
reproduction number, $R_{e},$ similarly.  
From \eqref{3.9b} let 
\begin{equation*}
f_{i} = \begin{bmatrix}
\theta  I S \\
\theta(1 - k_{2} ) S \\ 
0
\end{bmatrix}
\end{equation*}
So that \begin{equation*}
\frac{\partial f_{i}}{\partial x_{j}}  =
 \begin{pmatrix}
0 & \theta S & 0 \\
0 & 0 & 0 \\
0 & 0 & 0
\end{pmatrix}.
\end{equation*}
and also 
\begin{equation*}
v_{i} = \begin{bmatrix}
\beta L + \mu L + \eta(1 + k_{1} )L \\
\mu I + \gamma I + \alpha(1 + k_{2} )I - \beta L 
- \theta S(1 -k_{2} )\\
\mu T + \omega T - \eta(1 + k_{1} )L - \alpha(1 + k_{2} )I
\end{bmatrix}
\end{equation*}
so that 
 \begin{equation*}
\frac{\partial v_{i}}{\partial x_{j}} E_{0} = V =
 \begin{pmatrix}
(\beta + \mu + \eta(1 + k_{1} )) & 0 & 0 \\
- \beta & (\mu + \gamma + \alpha(1 + k_{2} )) & 0 \\
- \eta(1 + k_{1} ) & - \alpha(1 + k_{2} ) & (\mu + \omega)
\end{pmatrix}.
\end{equation*}

The co-factors matrix is  then 
\begin{equation*}
\small{
\begin{pmatrix}
(\mu + \gamma + \alpha(1 + k_{2} ))(\mu + \omega) & 
-\beta(\mu + \omega) &
\alpha \beta(1 + k_{2} ) + \eta(1 + k_{1} )(\mu + \gamma
 + \alpha(1 + k_{2} ) \\
0 & (\beta + \mu + \eta(1 + k_{1} ) )(\mu + \omega) & 
- \alpha(1 + k_{2} )(\beta
+ \mu + \eta(1 + k_{1} )) \\
0 & 0 & (\beta + \mu + \eta(1 + k_{1} )(\mu + \gamma 
+ \alpha(1 + k_{2} ))
\end{pmatrix}}.
\end{equation*}
So, following the same  procedure for the computation of 
reproduction number, we get the eigenvalues: 
\[ \lambda^{2} = 0 ~~ \text{or} ~~
\lambda = \frac{(\beta \theta S )}{(\beta 
+ \mu + \eta(1 + k_{1} ))(\mu + \gamma + \alpha(1 + k_{2} ))} . \]

Therefore, the effective reproduction number
\begin{equation}\label{effrepno}
 R_{e} = \frac{(\beta \theta S )}{(\beta + \mu 
 + \eta(1 + k_{1} ))(\mu + \gamma + \alpha(1 + k_{2} ))} .   
\end{equation}
 
It can be  observed here  that the effective reproduction number 
is far less than the basic reproduction number  
i.e. $R_{e} <<R_{0}$ since 
\begin{equation*}
  \frac{(\beta \theta S )}{(\beta + \mu + \eta(1 + k_{1} ))(\mu
   + \gamma + \alpha(1 + k_{2} ))}
    <  \frac{(\beta \theta S)}{(\beta + \mu 
    + \eta)(\mu + \gamma + \alpha)} .
\end{equation*}
Using the data in Table \ref{dt}, we see that 
\begin{eqnarray*}
R_{e} &=& \frac{\sigma \beta \theta((\mu + \upsilon) 
- \mu f)}{\mu (\mu + \alpha + \gamma) (\mu + \beta + \eta)
(\mu + \upsilon)} < 
\frac{(1 - k_{2} )(\alpha + \gamma + \mu)(\beta + \mu + \eta)}
{(\alpha(1 + k_{2} ) + \gamma + \mu)(\eta(1 + k_{1} ) 
+ \beta + \mu) ) } \\
 &=&  0.09700176367 < 0.1566070960 \ll 1.  
 \end{eqnarray*}
This gives the result. 
\end{proof}

One can as well use the Routh-Hurwitz criteria, see \cite{Ver},   to 
assess the local stability of the model. In this technique, an 
equilibrium point is called asymptotically stable if 
the trace of the Jacobian of the next generation matrix of the model 
is less than zero. This means that all the roots of the characteristic 
polynomial have negative real parts. We have the next result. 
\begin{theorem}\label{T3}
The extended model, \eqref{3.9b}, is asymptotically stable.
\end{theorem}
\begin{proof}
From the model system (\ref{3.9}) we set   
\begin{eqnarray*} 
f_{1} & = & f \sigma - (\upsilon  + \mu )Q \\
f_{2} & =& \upsilon  Q + (1 - f) \sigma - \theta(1 - k_{2} )S 
+  \delta R -
 \mu S - \theta I S \\
f_{3} & = & \theta SI - (\beta +  \mu  + \eta(1 + k_{1} ) )L \\
f_{4} & = & \beta L  + \theta S (1 - k_{2} )  
- (\mu + \gamma + \alpha(1 + 
 k_{2} ) )I \\
f_{5} & = & \omega T - (\mu  + \delta )R \\
f_{6} & = & \eta(1 + k_{1} )L + \alpha(1 + k_{2} )I 
- \mu T - \omega T .
\end{eqnarray*}
At the disease-free equilibrium, the Jacobian matrix, $J|_{(DFE)},$ 
 of the model and the associated characteristics determinant 
 $|J -\lambda I| =0$  implies  
\begin{eqnarray*}
 && ( -(\mu + \upsilon) - \lambda)(-(\mu + \theta(1 -k_{2} ))
 - \lambda)( -(\beta + \mu + \eta(1 + k_{1} )
- \lambda )( -\alpha(1 + k_{2} ) \\ 
&& +  \gamma + \mu) - \lambda)( -(\mu + \delta)- \lambda)(
  -(\mu + \omega) - \lambda) = 0 .
\end{eqnarray*}
Therefore, 
\begin{eqnarray*}
\lambda_{1} & =& -(\mu + \upsilon), \\
\lambda_{2} & = & -(\mu + \theta(1 -k_{2} )) ,\\
\lambda_{3} & = & -(\beta + \eta(1 + k_{1} ) + \mu ) ,\\
\lambda_{4} & = & -( \gamma + \alpha(1 + k_{2} ) + \mu ) ,\\
\lambda_{5} & = & -(\mu + \delta) ,\\
\lambda_{6} & = & -(\mu + \omega) ;
\end{eqnarray*}
where $\lambda_{1} , \lambda_{2} , \lambda_{3} , 
\lambda_{4} , \lambda_{5}$ 
and   $\lambda_{6}$ corresponding to $\lambda_{j} , j = 1, 2, \cdots 6$
are the eigenvalues. Since all the eigenvalues are negative, 
we conclude that the model is asymptotically stable.
\end{proof}

We observe that this result indicates that  the control interventions 
in the model such as the use of condom, education
enlightenment programme and treatment are effective enough to almost 
wipe out the disease in a limited time.

One can also discuss the global stability of the Model.
Here, we use Lyapunov  direct method, (see e.g. \cite{Ver}), 
to derive sufficient condition for the global stability of the system.

\begin{theorem}\label{defs}\cite{Ver}. 
Let $D$ be an open subset of  $\mathbb{R}^6$ and 
$\dot{x} = f(x)$ a system of  differential equations for  $x \in D.$
Let $V : D \rightarrow  \mathbb{R}^6$ be a smooth function. 
Suppose $V$ is positive definite around a critical point 
$x^{\ast} \in D,$ i.e., $V(x^{\ast} ) = 0$ and 
$V(x) > 0,$ for all $x \in D$ with $x \neq x^{\ast} .$ 
Then, the critical solution $x = x^{\ast}$ is asymptotically stable 
in the sense of Lyapunov if the time derivative $\dot{V} < 0.$
\end{theorem}
For Proof, see e.g. \cite{Ver}. We can now present one of the main 
results of this study.

\begin{theorem} \label{4} 
The equilibrium point solution is globally asymptotically stable. 
\end{theorem}

\begin{proof}
In 
 $\displaystyle{R_{0} = \frac{\beta \theta S}{(\alpha 
 + \gamma + \mu)(\beta + \eta + \mu )}},$  
let $\displaystyle{\mu_{0} = \alpha + \gamma + \mu}$ and 
$\displaystyle{\mu_{1} = \beta + \eta + \mu }.$

Then $\displaystyle{R_{0} = \frac{\beta \theta S}{\mu_{0} \mu_{1}} }$
so that  $\displaystyle{\frac{\mu_{0} \mu_{1} R_{0}}{\beta} 
= \theta S }.$
Since  $\displaystyle{R_{0} < 1},$ then  $\displaystyle{R_{0}
 - 1 < 0}$ and so $\displaystyle{\frac{\mu_{0} \mu_{1}}{\beta} 
 (R_{0} - 1) < 0 }.$

If we choose $\displaystyle{\varepsilon > 0}$ sufficiently small 
such that $\sigma  = \frac{\mu_{1}}{\beta} - \varepsilon$ 
where $\frac{\beta + \eta + \mu}{\beta}$
is the mean number of susceptible people that are infected by one 
infectious individual during the infectious time and $\sigma$ 
is the recruitment rate direct to susceptible class.  
 
Now, define a Lyapunov function $V(L ,I)$ by 
\begin{equation}\label{Lyap2}
V = L + \sigma I .
\end{equation}
Clearly,  $V$ is a positive definite function and so we can show 
 that its time derivative  $\dot{V}$ is negative definite. That is, 
 \begin{eqnarray*}
\frac{dV}{dt} &=& \frac{dL}{dt} + \frac{dI}{dt} \sigma \\
&=& \theta I S - (\beta + \eta + \mu )L + \sigma(\beta L -(\alpha 
+ \mu + \gamma)I ) \\
&=& \theta I S - \mu_{1} L + \sigma \beta L - \sigma \mu_{0} I <0.
\end{eqnarray*}
Take co-efficients of $L$ and $I$  to have that 
\begin{eqnarray*}
L_{c} &=& \sigma \beta - \mu_{1} \\
I_{c} &=& \theta S - \mu_{0} \sigma .
\end{eqnarray*}
From $L_{c}$ we have that  
\begin{eqnarray*}
(\frac{\mu_{1}}{\beta} - \varepsilon )\beta - \mu_{1} = 
- \varepsilon \beta < 0 \Rightarrow 
 \theta S - \mu_{0} (\frac{\mu_{1}}{\beta} - \varepsilon ) 
 &=& \theta S - \frac{\mu_{0} \mu_{1}}{\beta} +\varepsilon \mu_{0} \\
 \frac{\mu_{0} \mu_{1} R_{0}}{\beta} - \frac{\mu_{0} \mu_{1}}{\beta} +
 \varepsilon \mu_{0} 
&=& \frac{\mu_{0} \mu_{1}}{\beta} (R_{0} - 1) + 
\varepsilon \mu_{0} < 0 . 
\end{eqnarray*}

This implies, 
\[\frac{dV}{dt} = \frac{dL}{dt} + \sigma \frac{dI}{dt}  < 0.\]
Hence, $\dot{V} < 0$ as required.
\end{proof}

Finally, we conduct the sensitivity analysis of the model on the 
effective reproduction number.

\begin{theorem}\label{mainmain}
The dynamics of the model is variable and dependent on the 
waning rate, the control parameters and the interaction between the 
latent and infected classes.  
\end{theorem}

\begin{proof}
Following \cite{Gar}, we prove Theorem (\ref{mainmain}) using the 
sensitivity index function of the effective reproduction number given 
by  
\begin{equation}\label{senseindex}
\gamma_{R_{e}}(\cdot)  = \frac{\partial R_{e}}{\partial (\cdot )}  
\times \frac{(\cdot )}{R_{e}}.
\end{equation}
So, the variation of the treatment rate of the infected class is 
\[\gamma_{R_{e}}(\alpha ) = 
\frac{\alpha}{R_{e}} \times \frac{\partial}{\partial \alpha} 
R_{e}  = -(\frac{\alpha((\beta + \mu + 
\eta(1 + c))(1 + k))}{(\beta + \mu + \eta (1 + c))(\mu 
+ \Omega + \alpha (1 + k))} ).\]
Now taking  $\Omega = 0.01, \alpha = 0.2, \beta = 0.01, 
c = 0.5, \eta = 0.1, 
k = 0.8, \mu = 0.2],$ we have  
\[\gamma_{R_{e}}(\alpha ) = -0.6315789474. \]
Similarly, the variation of the rate at which the the recovered class becomes susceptible again is
\[\gamma_{R_{e}}(\delta ) = 
\frac{\delta}{R_{[{e}]}} \times \frac{\partial}{\partial \delta} 
R_{[{e}]}  = \lambda( (\frac{\alpha((\beta + \mu + \upsilon(1 + f))(1 + 
k))}{(\beta + \mu + \upsilon (1 + f))(\mu + \delta 
+ \theta (1 + k))} )).\]
Take $\delta = 0.4, \theta = 0.5, \beta = 
0.01, f = 0.91, \upsilon = 0.4, k = 0.8, \mu = 0.2$ to have  
\[\gamma_{R_{e}}(\delta )= +1.000000000 .\]
Continuing this way, we obtain the sensitivity index table of the model 
as Table  \ref{t22}:
\begin{table}[!!ht]
\centering
  \resizebox{0.8\textwidth}{!}{\begin{minipage}{\textwidth}
 \centering
\begin{tabular}{ |l |c| c| r||}
\hline
Parameters  & Sensitivity indices $\gamma_{R_{e}}(\cdot )$  \\
\hline
$\theta$  & $+1$ \\
\hline
$\sigma $ &   $+1$  \\
\hline
$\beta$ & $+ 0.97222222218$\\
\hline
$\upsilon $ &   $+0.2902711325$  \\
\hline
$\gamma $ &   $-0.01754385965$  \\
\hline
$ k_{1} $ &   $-0.1388888889$  \\
\hline
$\mu $ &   $-2.176703880$  \\
\hline
$\eta $ &   $-0.4166666667$  \\
\hline
$\alpha $ &   $-0.6315789474$  \\
\hline
$f $ &   $-0.4354066982$  \\
\hline
$k_{2} $ &   $-0.4280701754$  \\
\hline
\end{tabular}
\caption{Sensitive parameters and values.}
\label{t22}
\centering
\end{minipage}}
\end{table} \\

Graphically, we have the effect of increasing waning rate on the 
susceptible and immune classes as Figure \ref{image1} while Figure \ref{image2} shows the effect of a higher waning rate.
\begin{figure}[hbtp]
\centering
\includegraphics[scale=0.6]{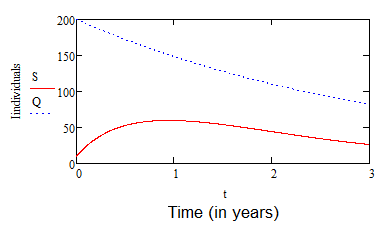}
\caption{Waning rate $\upsilon = 0.2$. } 
\label{image1}
\end{figure}
\begin{figure}[hbtp]
\centering
\includegraphics[scale=0.6]{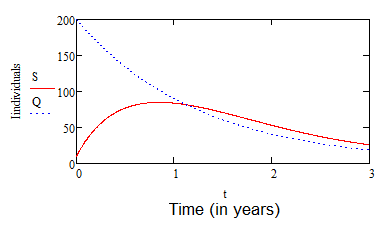}
\caption{Waning rate  $\upsilon = 0.6.$}
\label{image2}
\end{figure}
The graphs indicate that when the waning rate $\upsilon$ is 
low (i.e., $\upsilon = 0.2$), the passive immune population decreases 
exponentially with time, while when the waning rate is high, 
(i.e., $\upsilon = 0.6$), the passive immune 
population decreases faster  and varnishes with time. The continuous 
decay in the population of the immune class (Q) with time is due to the 
fact that the immunity conferred on the individuals in this class is 
temporal and hence, expires with time. However, the susceptible 
population increases  slower to the turning 
point at about one year and three months as the waning rate 
$\upsilon$ is low and increases faster  as the waning rate 
$\upsilon$ is high as shown in Figure \ref{image1} and  Figure 
\ref{image2} respectively. In both cases, the susceptible class later 
decreases with time due to the interaction among the latent, infected 
and the susceptible classes coupled with the natural mortality rate 
$\mu .$

Next figures show the effect when the control measures are removed and 
when they are introduced.  Figure \ref{image4a} shows when $k_{1} = 
k_{2} = 0$ that the susceptible individuals first increased after $45$ 
days due to recruitment into it,  and the trajectory decreases with 
time  as more people are getting infected with the disease. 
\begin{figure}[hbtp]
\centering
\includegraphics[scale=0.6]{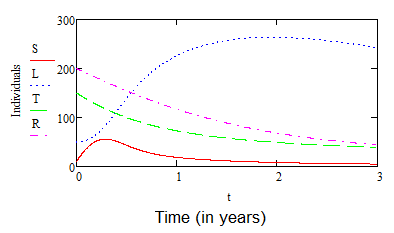} 
\caption{When $k_{1} = 0 = k_{2}.$}
\label{image4a}
\end{figure}\\
In the same way the latent and the infected 
population show exponential increase because more people are getting 
infected without control measure. However, the treated and the 
recovered population  show drastic exponential decay with time   as a 
result of no control measure.

Figure \eqref{image4b} suggests that susceptible population increases 
exponentially with time since more people get treated, recovered and 
join the susceptible class again because of increase in control 
measure. Also the recovered class increases with time for about $7$  
months as the control measure is high (i.e., $k_{1} = 0.7$ and 
$k_{2} =0.9$), and started decreasing again asymptotically with time 
because  recovery from gonorrhoea is with temporal immunity. However, 
the trajectory of the latent and the infected classes decreased  to 
zero with time due to increase in control measure.

\begin{figure}[hbtp]
\centering
\includegraphics[scale=0.6]{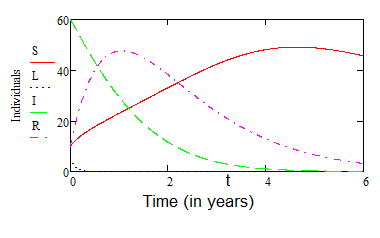}
\caption{Effect when $k_{1} = 0.7$ and $k_{2} =0.9$.}
\label{image4b}
\end{figure} 

Figure (\ref{image5}) indicates that when the interaction rate 
is low (i.e., $\theta =0.3$), the latent and the infected 
classes decrease exponentially with time, and even varnishes in the 
long run since there will be almost nobody to contact  and suffer 
the disease. It is also shown that when the interaction rate 
$\theta = 0$, the basic reproduction number of the disease becomes 
zero.  
\begin{figure}[hbtp]
\centering
\includegraphics[scale=0.6]{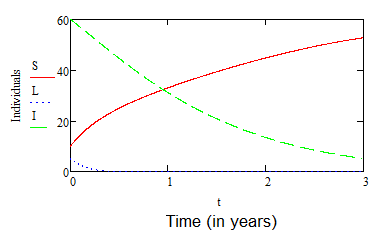}
\caption{Effect when  $\lambda = \theta I, ~~ \theta =0, ~~ 
\text{and when}~~  \theta =0.3.$}
\label{image5}
\end{figure} \\
We conclude that  
\[\lim_{\theta \rightarrow 0} R_{0} =
\lim_{\theta \rightarrow 0} \frac{(\beta \theta S)}{(\beta + \mu + \eta)(\mu + 
\gamma + \alpha)} = 0.\] 
At this point, the contact rate $\lambda$ becomes zero and hence, 
nobody suffers the disease. 
\end{proof}

\section{ Conclusion}
Based on the analysis and results of this work, we observed 
that the solution set of the state variables of the extended 
model are all positive for  $t\ge 0.$ This clearly showed that the 
population we are studying is not a zero population.  We obtained 
an important threshold parameter called the effective reproduction 
number $\mathbb{R}_{e}$ using the next generation matrix method. 
The result of the effective reproduction number is less than one. 
The local and global stability were investigated using Routh-Hurtz 
criteria and Lyapunov method respectively and both were asymptotically 
stable for $R_{e} < 1.$ 

From the graphical illustrations, we concluded that immune population 
continues to decay exponentially due to temporal immunity conferred 
on the individuals in the immune class. We also concluded that 
reproduction number of the infection grows when 
there is no control measure in the model and decays when control 
measure is applied in the model. Also from the analysis, it can be 
seen that the disease can  be totally eliminated from the community, 
because the sensitivity index shows that the lower the waning 
rate $\nu ,$ the more the exponential decrease in the passive immunity 
but the susceptible population increases  to the equilibrium and 
wanes asymptotically due to the presence of the control parameters 
and restricted interaction of the latent and infected classes. 
     
The sensitivity analysis of the effective reproduction 
number further shows that the rate of expiration of immunity 
$\upsilon,$ coupled with the interaction rate $\theta,$ recruitment 
rate $\sigma$ and the infectious rate $\beta$  are the most sensitive 
parameters to be targeted by way of intervention strategy.

\end{document}